\documentclass[lettersize,journal]{IEEEtran}
\usepackage{amsmath,amsfonts,dsfont}
\usepackage{algorithmic}
\usepackage{array}
\usepackage[caption=false,font=normalsize,labelfont=sf,textfont=sf]{subfig}
\usepackage{textcomp}
\usepackage{stfloats}
\usepackage{url}
\usepackage{verbatim}
\usepackage{graphicx}
\usepackage{xcolor}
\usepackage{caption}
\usepackage{amsthm}
\usepackage[ruled,vlined]{algorithm2e}
\usepackage{makecell}

\usepackage[numbers,sort&compress]{natbib}

\def\BibTeX{{\rm B\kern-.05em{\sc i\kern-.025em b}\kern-.08em
    T\kern-.1667em\lower.7ex\hbox{E}\kern-.125emX}}
\usepackage{balance}

\def\x{{\mathbf x}}
\def\X{{\mathbf X}}

\def\w{{\mathbf w}}
\def\z{{\mathbf z}}
\def\r{{\mathbf r}}
\def\d{{\mathbf d}}
\def\B{{\mathbf B}}
\def\Lam{{\mathbf \Lambda}}

\newcommand{\overbar}[1]{\mkern 1.5mu\overline{\mkern-1.5mu#1\mkern-1.5mu}\mkern 1.5mu}
\newcommand{\RN}[1]{%
  \textup{\uppercase\expandafter{\romannumeral#1}}%
}

\newtheorem{theorem}{Theorem}[section]

\begin{document}


\title{Minimax Concave Penalty Regularized Adaptive System Identification}
\author{Bowen Li, \emph{Student Member, IEEE}, Suya Wu, \emph{Student Member, IEEE}, Erin E. Tripp, Ali Pezeshki, \emph{Member, IEEE}, and Vahid Tarokh, \emph{Fellow IEEE}
\thanks{This work is supported in part by the AFOSR under award FA8750-20-2-0504. The work of E. E. Tripp was supported in part by AFOSR grant 21RICO035. Any opinions, findings and conclusions or recommendations expressed in this material are those of the authors and do not necessarily reflect the views of the U.S. Air Force Research Laboratory. Cleared for public release 13 September 2022: Case number AFRL-2022-4343.\\ 
\indent B. Li and A. Pezeshki are with the Department of Electrical and Computer Engineering, Fort Collins, CO 80523 USA. (email: bowenli@colostate.edu; Ali.Pezeshki@colostate.edu) \\
\indent S. Wu and V. Tarokh are with the Department of Electrical and Computer Engineering, Duke University, Durham, NC 27708 USA. (email: suya.wu@duke.edu; vahid.tarokh@duke.edu) \\
\indent E. E. Tripp is with Air Force Research Laboratory, Rome, NY 13441 USA. (email: erin.tripp.4@us.af.mil)

}}


\markboth{IEEE TRANSACTIONS ON SIGNAL PROCESSING}
{How to Use the IEEEtran \LaTeX \ Templates}

\maketitle

\begin{abstract}
We develop a recursive least squares (RLS) type algorithm with a minimax concave penalty (MCP) for adaptive identification of a sparse tap-weight vector that represents a communication channel. The proposed algorithm recursively yields its estimate of the tap-vector, from noisy streaming observations of a received signal, using expectation-maximization (EM) update. We prove the convergence to a local optimum of the static least squares version of our algorithm and provide bounds for the estimation error. We study the performance of the recursive version numerically. Using simulation studies of Rayleigh fading channel, Volterra system and multivariate time series model, we demonstrate that our recursive algorithm outperforms, in the mean-squared error (MSE) sense, the standard RLS and the $\ell_1$-regularized RLS.
\end{abstract}


\section{Introduction}
\IEEEPARstart{S}{\lowercase{ystem}} identification from streaming data arises in various applications, including echo cancellation~\cite{cui2004improved,naylor2006adaptive} and multipath channel communications~\cite{bruckstein2009sparse,li2014improved,bajwa2010compressed}. In a wide range of problems the finite impulse response to be estimated is sparse, but the nonzero tap weights and their support vary with time. Estimating such time-varying sparse vectors requires the development of suitable adaptive filtering algorithms with sparse regularization.  

Least mean squares (LMS) and recursive least squares (RLS) are two of the most widely used adaptive filtering algorithms. However, they do not enforce a sparse structure on their estimate of the system impulse response. To fill this gap, several variants of LMS and RLS type algorithms that use some form of sparse regularization have been proposed in the literature. Specifically, different regularization penalties, including $\ell_0$ \cite{gu2009l_,luo2020steady}, $\ell_1$\cite{chen2009sparse}, and $\ell_p$ norms \cite{li2014improved}, along with reweighting of the zero attractor \cite{chen2009sparse}, were incorporated into the quadratic cost function of LMS algorithm. However, LMS type algorithms in general exhibit slow convergence and are prone to producing poor tracking performance in time-varying environment \cite{haykin2008adaptive}. 


A sparse adaptive filtering algorithm based on RLS and $\ell_1$-norm regularization, called SPARLS, is proposed in \cite{SparseRLS}. This algorithm 
benefits from the competitive computational complexity of the expectation-maximization (EM) algorithm and has guaranteed convergence to a global optimum due to the convexity of $\ell_1$ norm. An online nonlinear version of SPARLS is developed in  \cite{kalouptsidis2011adaptive} and applied to estimating nonlinear multipath channels. Subsequently, \cite{han2017slants} proposed a SPARLS algorithm with group LASSO regularization to model multivariate nonlinear time series data.

However, the LASSO estimator is biased and, with $\ell_1$ regularization, it tends to over-shrink variables that it retains. Minimax concave penalty (MCP), introduced by Zhang in \cite{zhang2010nearly}, is an alternative nonconvex penalty that performs variable selection while maintaining unbiasedness to a good extent. In \cite{wang2018admm}, Wang \textit{et al.} employ MCP in a quadratic programming problem for sparse recovery and solve the problem using the alternating direction method of multipliers (ADMM) and iterative hard thresholding. Their MCP regularized sparse recovery exhibits better performance than various state-of-the-art sparse signal reconstruction methods. However, their framework does not apply to online adaptive system identification, with time-varying tap-weight vectors, from streaming data.    
 


In this paper, we first introduce an MCP-regularized least squares algorithm for adaptive sparse system identification. The proposed algorithm, which we refer to as SPALS-MCP, uses expectation-maximization (EM) update to produce an estimate of the tap-weight vector. We prove the almost sure convergence to a stationary point of the static least squares version of our algorithm and provide bounds for the estimation error. We then introduce a recursive version of the above algorithm, called SPARLS-MCP, and study its performance numerically. Using simulation of Rayleigh fading channel, Volterra system, and multivariate time series model, we demonstrate that our SPARLS-MCP algorithm outperforms, in the mean-squared error (MSE) sense, the standard RLS and the $\ell_1$-regularized RLS.



\section{Mathematical Preliminaries}
\subsection{Problem Formulation}
At instant time $i \; (i = 1,\cdots,n)$, $\x(i) = [x_{i-M+1}, \cdots, x_{i}]^{T}$ is the input signal vector of recent $M$ signals and $\w(i)=[w_{1,i}, \cdots, w_{M, i}]^{T}$ is the tap-weight vector of length $M$. The desired response, denoted by $d(i)$, is composed of the filter output $y(i)$ and noise $\epsilon(i)$. Specifically, 
$$
d(i) = y(i) + \epsilon(i) = \w^{H}(i) \x(i) + \epsilon(i)
$$
where $(.)^{H}$ represents the conjugate transpose and $\epsilon(i) \; (i=1,\cdots,n)$ are i.i.d Gaussian noises with variance $\sigma^2$. 

Define the instantaneous error at instant time $i$ of the filter as $e(i) := d(i)-\hat{y}(i) = d(i)-\hat{\w}^{H}(i)\x(i)$, where $\hat{\w}(i)$ is the estimated tap-weight vector at instant time $i$. At instant time $n$, the objective of adaptive algorithms is to minimize the weighted sum of square loss $\sum_{i=1}^{n} \lambda^{n-i}|e(i)|^2$, where $\lambda$ is a hyperparameter that referred to as the forgetting factor.

We assume that the unknown tap-weight vector $\w(n)$ is $k$-sparse and that both its support and non-zero values can vary over time. We henceforth introduce a regularization term for sparsity promotion, giving rise to minimizing 
\begin{equation}
\label{eq:ob_fc}
    \frac{1}{2\sigma^2}\sum_{i=1}^{n} \lambda^{n-i}(d(i)-\w^{H}(n)\x(i))^2+\rho(\w(n), \gamma)
\end{equation}
\sloppy In Eq. (\ref{eq:ob_fc}), the first term stands for the RLS cost, while the second term $\rho(\w(n), \gamma) = \sum_{j=1}^{M} \rho\left(\left|w_{j, n}\right|, \gamma\right)$ is a general penalty function dictated by a non-negative $\gamma$ \cite{shen2019structured}. Many functions are available to characterize the penalty. For example, Eq. (\ref{eq:ob_fc}) represents $\ell_0$ regularized RLS when $\rho(\w(n), \gamma) = \gamma \| \w(n)\|_{0}$. However, directly solving the $\ell_0$ norm regularized problem is computationally intractable due to the combinatorial nature of placing up to $k < M$ non-zero elements in $M$ dimension in $\w(n)$. To address this issue, the SPARLS algorithm with convex relaxation is proposed in \cite{SparseRLS} by setting $\rho(\w(n), \gamma) = \gamma \|\w(n)\|_{1}$. It solves an alternative convex optimization problem that is computationally feasible and results in sparse solutions. In the rest of paper, we denote $\ell_1$ regularized RLS as the SPARLS-$\ell_1$ algorithm. 

For notational convenience, we rewrite Eq. (\ref{eq:ob_fc}) into a compact form as: 
\begin{equation}
\label{eq:ob_fc_mat}
    \frac{1}{2\sigma^2}\| \Lam^{1/2}(n)\d(n)-\Lam^{1/2}(n)\X(n) \overbar{\w}(n)\|_{2}^{2}+\rho(\w(n), \gamma)
\end{equation}
where 
$$
\begin{aligned}
\Lam(n) & := \text{diag}(\lambda^{n-1},\lambda^{n-2},\cdots,1) \\ 
\d(n) & := [d(1),d(2),\cdots,d(n)]^{T} \\
\X(n) & := [\x(1),\x(2),\cdots,\x(n)]^{T}
\end{aligned}
$$ 
$\Lam^{1/2}(n)$ is a diagonal matrix with diagonal elements being $\sqrt{\Lam_{ii}(n)} \; (1\leq i \leq n)$ and $\overbar{\w}(n)$ is the conjugate of vector $\w(n)$. 

\subsection{Minimax Concave Penalty (MCP)}
In contrast to the convex relaxation of $\ell_1$ norm, the minimax concave penalty (MCP), first introduced for unbiased variable selection \cite{zhang2010nearly} and further generalized to a class of semiconvex sparsity promoting functions \cite{shen2019structured}, is a type of nonconvex constraints. 
Specifically, we have the scaled MCP that is defined by 
\begin{equation}
    \rho_{\alpha}(w)=|w|-\operatorname{env}_{\alpha}|w|= \begin{cases}|w|-\frac{1}{2 \alpha} w^{2}, & \;|w| \leq \alpha \\ 
    \frac{1}{2} \alpha, & \text { o.w. }\end{cases}
\end{equation} 
where \[ \operatorname{env}_{\alpha}|w| = \begin{cases} \frac{1}{2 \alpha} w^{2}, & \;|w| \leq \alpha \\ 
    |w| - \frac{1}{2}\alpha, & \text { o.w. }\end{cases} \]
is the Moreau envelope function and $\alpha$ is the parameter characterizing the envelope height \cite{shen2019structured}. 

\begin{figure}[h!]
    \centering
    \includegraphics[width=6cm]{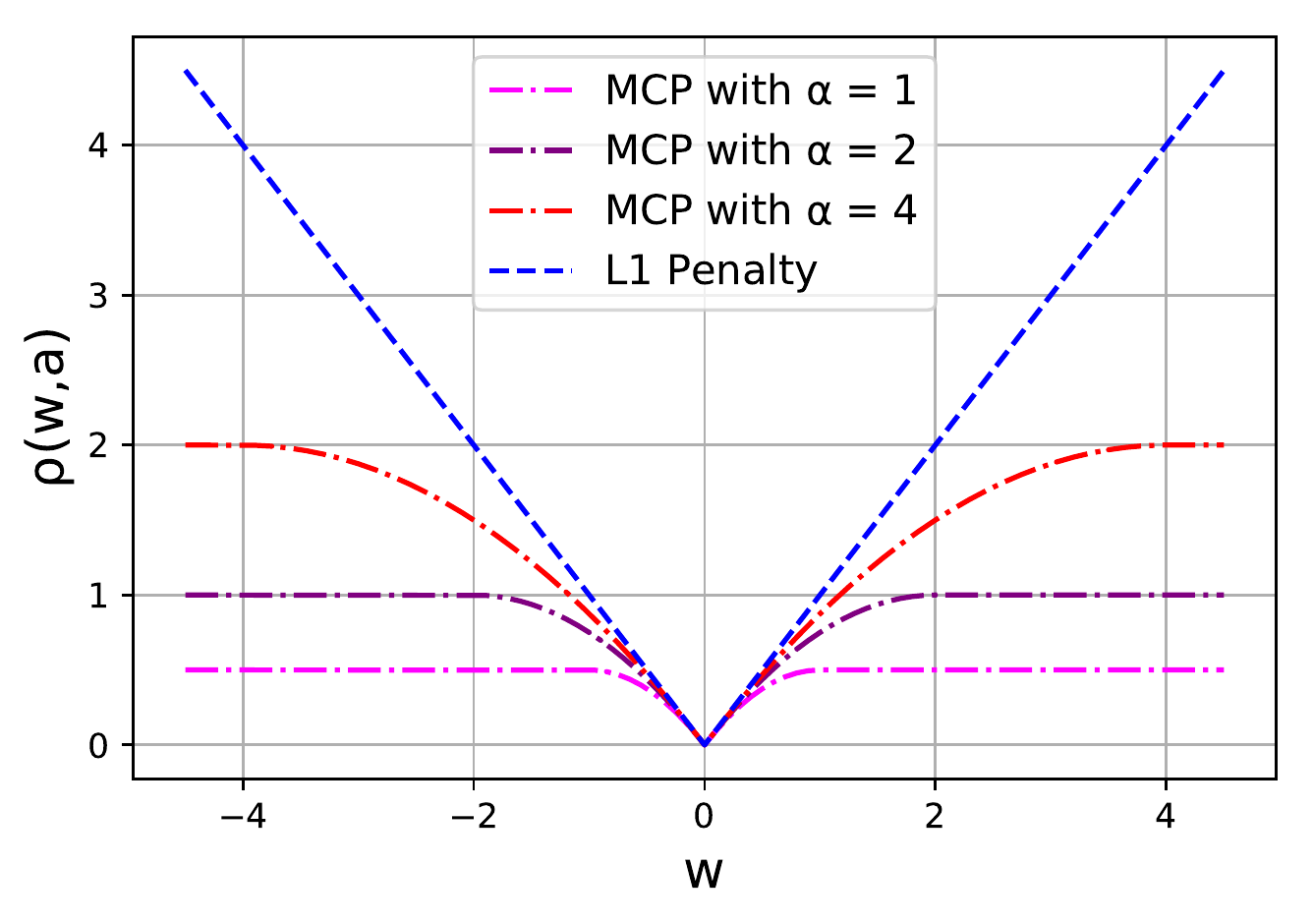}
    \caption{The Comparison between scaled MCP and $\ell_1$ Penalty} 
    \label{mcp_demo}
\end{figure}

Fig. \ref{mcp_demo} compares the shape of scaled MCP to $\ell_1$ penalty. It shows that MCP approximates $\ell_1$ penalty when $\alpha$ is arbitrarily large. In particular, MCP maintains the shape of the $\ell_1$ norm around the origin while approximating the scaled $\ell_0$ norm away from the origin. This structure leads to thresholding behavior in the iterative updates of the EM algorithm, which is captured by the proximal operator.

The proximal operator of scaled MCP regularized least squares is:
\begin{equation*}
\label{eq:mcp_prox}
    \operatorname{prox}_{\beta \text{,} \rho_\alpha}(r) := \arg\min_{w \in \mathbb{C}} ~\frac{1}{2\beta} (r - w)^{2} + \rho_{\alpha}(w)
\end{equation*}
The explicit solution of $\operatorname{prox}_{\beta \text{,} \rho_\alpha}(r)$ depends on the relative values of $\alpha$ and $\beta$. 

\noindent If $\beta<\alpha$, 
\begin{equation*}
\label{eq:sol1}
   \operatorname{prox}_{\beta \text{,} \rho_\alpha}(r) =\left\{
    \begin{array}{ll}
        0,   & |r| \leq \beta \\
        \frac{\alpha}{\alpha-\beta}(1-\frac{\beta}{|r|})r,   & \beta< |r| \leq \alpha \\
        r,   & |r| > \alpha
    \end{array}\right.
\end{equation*}
which is the firm thresholding operator. 

\noindent If $\beta = \alpha$,
\begin{equation*}
\label{eq:sol2}
    \operatorname{prox}_{\beta \text{,} \rho_\alpha}(r)=\left\{
    \begin{array}{lll}
        0,   & |r|<\alpha \\
        {[0, \alpha],}   & |r|=\alpha \\
        r,   & |r|>\alpha
    \end{array}\right.
\end{equation*}
where $[0, \alpha]$ represents a uniform selection from the values between $0$ and $\alpha$.

\noindent If $\beta>\alpha$, then the hard thresholding operator is obtained,
\begin{equation*}
\label{eq:sol3}
    \operatorname{prox}_{\beta \text{,} \rho_\alpha}(r)=\left\{
    \begin{array}{lll}
        0,   & |r|<\sqrt{\alpha \beta} \\
        \{0, r\},   & |r|=\sqrt{\alpha \beta} \\
        r,   & |r|>\sqrt{\alpha \beta}
    \end{array}\right.
\end{equation*} 
where $\{0, \alpha\}$ represents a uniform selection from one of $0$ and $\alpha$.

\begin{figure}[h!]
    \centering
    \includegraphics[width=0.5\textwidth]{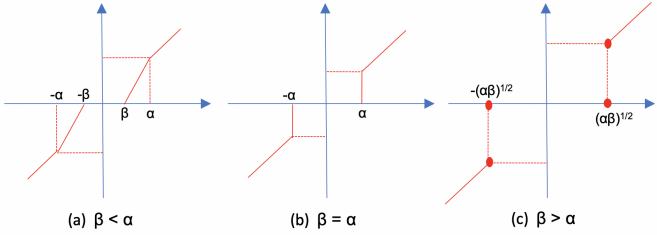}
    \caption{Proximal Operator of Scaled MCP} 
    \label{mcp_prox}
\end{figure}

Fig. \ref{mcp_prox} illustrates the proximal operators when $\beta<\alpha$, $\beta=\alpha$ and $\beta>\alpha$, respectively. This proximal operator can naturally be extended to the case with vector $\r \in \mathbb{C}^{M}$ in which we have:
\begin{equation*}
\label{eq:mcp_prox_vec}
    \begin{aligned}
    \operatorname{prox}_{\beta \text{,} \rho_\alpha}(\r) & := \min_{\w} \frac{1}{2\beta} \|\r - \w\|_{2}^{2} + \rho_{\alpha}(\w) \\
    & = \operatorname{prox}_{\beta \text{,} \rho_\alpha}(r_1) \times \cdots \times \operatorname{prox}_{\beta \text{,} \rho_\alpha}(r_M)
    \end{aligned}
\end{equation*}
where $\times$ is the Cartesian product operator. 

We then incorporate scaled MCP into the weighted least squares term to get the objective function for time $n$:
\begin{equation}
\label{eq:ob_fc_mcp}
    \min_{\w(n)} \frac{1}{2\beta}\| \Lam^{1/2}(n)\d(n)-\Lam^{1/2}(n)\X(n) \overbar{\w}(n)\|_{2}^{2}+\rho_{\alpha}(\w(n)) 
\end{equation}
where $\beta = \sigma^2\gamma$ is the corresponding penalty parameter that replaces $\gamma$ in Eq. (\ref{eq:ob_fc}) and let the regularization term $\rho_{\alpha}(\w,\gamma)$ become the scaled version $\rho_{\alpha}(\w)$. The time index of the objective function in (\ref{eq:ob_fc_mcp}) will have an increment of one after the new input and output signal arrive. The algorithm proposed in Sec. \RN{5} is designed to adaptively estimate the tap-weight vector at time $n$ based on the estimate from time $n-1$ and newly arrived signal. 

\section{Sparse Least Square with MCP}
The objective function in Eq. (\ref{eq:ob_fc_mcp}) involves a nonconvex penalty $\rho_{\alpha}(\w(n))$, which implies that Eq. (\ref{eq:ob_fc_mcp}) cannot be solved by convex optimization tools. 
We henceforth adopt a probabilistic approach with noise decomposition and maximum a posteriori (MAP) estimate produced by EM algorithm. This technique was first proposed for wavelet-based image restoration \cite{figueiredo2003algorithm}, and then applied to SPARLS-$\ell_1$ algorithm and its several variants \cite{SparseRLS,kalouptsidis2011adaptive,han2017slants}.

By setting up the normality assumption on desired signals $\d(n)$ and a proper prior on tap-weight vector $\w(n)$, we have $\d(n)=\X(n)\overbar{\w}(n) +\boldsymbol{\eta}(n)$ where $\boldsymbol{\eta}(n) \sim \mathcal{N}(\bold{0},\sigma^2\Lam^{-1}(n))$ and $p(\w(n))\propto \exp\{-\gamma\rho_{\alpha}(\w(n))\}$. Then the log of the posterior distribution of $\w(n)$ will be: 
\begin{equation}
\label{eq:w_post}
\begin{aligned}
    & \log\{p(\w(n)|\d(n))\} =  \log\{ \frac{p(\d(n)|\w(n)) p(\w(n)) }{p(\d(n))} \} \\ 
    & = \log\{ p(\d(n)|\w(n))\} + \log\{p(\w(n))\} -\log\{p(\d(n))\} \\
    & = -\frac{1}{2\sigma^2}\| \Lambda^{1/2}(n)\d(n)-\Lambda^{1/2}(n)\X(n) \overbar{\w}(n)\|_{2}^{2} \\
    & \hspace{4mm} - \gamma\rho_{\alpha}(\w(n)) + C,
\end{aligned}
\end{equation}
where $C$ accounts for the constant terms (relative to $\w(n)$). Hence, finding the maximum a posteriori (MAP) estimate of $\w(n)$ under our assumptions is equivalent to solving Eq. (\ref{eq:ob_fc_mcp}).

To maximize the log likelihood minus the regularization term, i.e.,
\begin{equation}
\label{eq:ob_map}
    \max_{\w(n)}\log p(\d(n)|\w(n)) -\gamma\rho_{\alpha}(\w(n))
\end{equation}
we first decompose the output signal into two parts, such that 
\begin{equation*}
    \begin{aligned}
        \z(n) & = \overbar{\w}(n) + \xi\boldsymbol{\eta}_{1}(n) \\
        \d(n) & = \X(n)\z(n)+ \Lam^{-1/2}(n)\boldsymbol{\eta}_{2}(n)
    \end{aligned}
\end{equation*}
where $\boldsymbol{\eta}_{1}(n) \sim \mathcal{N}(\bold{0},\bold{I}_{n})$ and $\boldsymbol{\eta}_{2}(n) \sim \mathcal{N}(\bold{0},\sigma^{2}\bold{I}_{n}-\xi^{2}\Lam^{1/2}(n)\X(n)\X^{H}(n)\Lam^{1/2}(n))$ are independent with each other. 
Next, we treat $\z(n)$ as the latent data and use EM algorithm to find the maximizer of penalized log-likelihood in Eq. (\ref{eq:ob_map}). Notice that the covariance matrix $\sigma^{2}\bold{I}_{n}-\xi^{2}\Lam^{1/2}(n)\X(n)\X^{H}(n)\Lam^{1/2}(n)$ needs to be semi-positive definite. Hence, we should have $\xi^{2} \leq \sigma^2/\lambda_1$ where $\lambda_1$ is the largest eigenvalue of $\Lam^{1/2}(n)\X(n)\X^{H}(n)\Lam^{1/2}(n)$. 

The $k^{\text{th}}$ iteration of EM algorithm at instant time $n$ is: \\

\textbf{E-step}: Calculate the expected log-likelihood given the previous estimate $\hat{\w}^{(k-1)}(n)$,
\begin{equation}
    \label{em-e}
    \begin{aligned}
    Q(\w(n)|\hat{\w}^{(k-1)}(n)) & = \mathbb{E}_{\z|\d,\hat{\w}^{(k-1)}(n)}\left[ \log p(\d(n),\z(n) | \w(n)) \right] \\
    & = - \frac{1}{2\xi^2}||\r(n)-\w(n)||_2^2
    \end{aligned}
\end{equation}
where $\r(n) = \left(\mathbf{I}-\frac{\xi^{2}}{\sigma^{2}} \X^{H}(n) \Lam(n) \X(n)\right) \hat{\w}^{(k-1)}(n) +\frac{\xi^{2}}{\sigma^{2}} \X^{H}(n)\Lam(n)\d(n)$. \\
    
\textbf{M-step}: Find the next $\w^{(k)}(n)$ that maximizes $Q(\w(n)|\hat{\w}^{(k-1)}(n))$ with penalty
\begin{equation}
    \label{em-m}
    \w^{(k)}(n) = \arg\max_{\w(n)} Q(\w(n)|\hat{\w}^{(k-1)}(n)) -\gamma \rho_{\alpha}(\w(n))
\end{equation}
    
The tap-weight vector $\w(n)$ may have ungrouped or grouped structure depending on the specific applications. The detailed algorithms to solve Eq. (4) at instant time $n$ for ungrouped and grouped version are as follows: 

\subsection{Ungrouped Version}
\begin{algorithm}

\caption{Sparse Least Squares with MCP (SPALS-MCP)}
\SetAlgoLined
    
 \textbf{Input}: $\X(n),\d(n),\Lam(n),\xi^{2}, \hat{\w}^{(0)}(n)$.\\
 ~\\
 \textbf{Calculate intermediate quantities:} \\
 $\B(n) = \mathbf{I}-(\xi^2/\sigma^2)\X^{H}(n) \Lam(n) \X(n)$ \\
 $\boldsymbol{\mu}(n) = (\xi^2/\sigma^2) \X^{H}(n)\Lam(n)\d(n)$ \\
 ~\\
 \textbf{Perform EM update:} \\
 \textbf{For $k$ in $1:K$:} \\ 
 \hspace{3mm} \textbf{E-step}: \\
 \hspace{3mm} $Q(\bold{w}|\hat{\bold{w}}^{(k-1)}(n)) = (-1/2\xi^2)\|\r(n) - \w\|_{2}^{2}$ \\
 \hspace{3mm} where $\r(n) = \B(n) \hat{\w}^{(k-1)}(n) + \boldsymbol{\mu}(n)$ \\

 \hspace{3mm} \textbf{M-step}: \\
\hspace{3mm} $
 \hat{\bold{w}}^{(k)}(n) = \arg\max_{\bold{w}} Q(\bold{w}|\hat{\bold{w}}^{(k-1)}(n)) - \gamma \rho_{\alpha}(\w) $ \\ \hspace{15.7mm} $= \operatorname{prox}_{\xi^{2}\gamma \text{,} \rho_\alpha}(\bold{r})
 $\\
 ~\\
 $\hat{\w}(n) = \hat{\w}^{(K)}(n)$ \\
 ~\\
 \textbf{Output}: Return $\hat{\bold{w}}(n)$

\end{algorithm}
Algorithm 1 presents the detailed steps for ungrouped version. 
$\xi^2$ is a predetermined hyper-parameter specified above. $\hat{\w}^{(0)}(n)$ is an initialization of $\w(n)$. The proximal operator $\operatorname{prox}_{\xi^{2}\gamma \text{,} \rho_\alpha}(\bold{r})$ follows what we have shown in Sec. \RN{2}-B.

\subsection{Grouped Version}
When the tap-weight vector exhibits group structure, we need to use group MCP for the regularization term. It is assumed that we have prior knowledge of group information. 
The group MCP penalty is obtained by setting 
\begin{equation*}
\label{eq:group_mcp}
    \rho(\w,\gamma) = \gamma \sum_{l=1}^{L}\rho_{\alpha}(\|\w_{l}\|_{2})
\end{equation*}
where $L$ is the number of groups and $\w_{l}$ is the $l^{\text{th}}$ group in tap-weight vector $\w$, i.e., $\w^{T} = (\w_{1}^{T},\cdots,\w_{L}^{T})$ and $\r^{T} = (\r_{1}^{T},\cdots,\r_{L}^{T})$

For the algorithm part, the only difference from Algorithm 1 is to replace the operation in M-step with group version proximal operator:
\begin{equation*}
\label{eq:prox_group_mcp}
    \hat{\w}_{l} = \left\{
    \begin{array}{lll}
        \frac{\alpha}{\alpha-\gamma\xi^2}(1-\frac{\gamma\xi^2}{\|\r_{l}\|_{2}})_{+}\r_{l},   & \|\r_l\|_{2} \leq \alpha \\
        \r_{l},   & \|\r_{l}\|_{2} > \alpha
    \end{array}\right. \; (1\leq l \leq L)
\end{equation*}
The derivation of this group version proximal operator is shown in Appendix A. 

\section{Analysis of Algorithm}
\subsection{Convergence}

\begin{theorem}
For a given time index $i$, the SPALS-MCP algorithm converges to a stationary point of the cost function in Eq. (\ref{eq:ob_fc_mcp}) if the firm thresholding is used. 
\end{theorem}

\begin{proof}
In the EM update of SPALS-MCP algorithm, both $Q(\w|\hat{\w}^{(k)}(i))$ and $\rho_{\alpha}(\w)$ are continuous with respect to $\w$ and $\hat{\w}^{(k)}(i)$. The firm thresholding operator employed in M-step is also continuous with respect to $\w$. Theorem 2 in \cite{wu1983convergence} shows that the EM algorithm converges to a stationary point of the cost function in Eq. (\ref{eq:ob_fc_mcp}) or (\ref{eq:w_post}) based on the above continuities. 
\end{proof}

\indent \textit{Remark 1:} 
The proximal operator in the EM update will be discontinuous at particular points if hard thresholding is used. However, we only have finitely many single point discontinuities if a finite number of EM iterations are implemented at each time instant. Hence, SPALS-MCP algorithm converges almost surely to a stationary point of the cost function in Eq. (\ref{eq:ob_fc_mcp}) when hard thresholding is used.  \\

\indent \textit{Remark 2:} We note that convergence of a class of iterative shrinkage/thresholding algorithms (ISTA), involving proximal operators of semi-convex regularizers, has been studied in \cite{kowalski2014thresholding}. Specifically, \cite{kowalski2014thresholding} shows that ISTA converges to a stationary point of a loss function, with a semi-convex regularizer, if the unregularized loss is $L$-Lipschitz differentiable. Accelerated algorithms for a broader class of problems that include nonconvex and nonsmooth penalties have been proposed in \cite{li2015accelerated}. Although MCP is in the class of regularizers considered in the above papers, our iterative algorithm and convergence analysis are different from those in \cite{kowalski2014thresholding,li2015accelerated}. We use the EM algorithm (instead of ISTA), and our convergence analysis follows that of the EM algorithm, which relies on the continuity of the expected log-likelihood instead of Lipschitz differentiability. 

\subsection{Estimation Error}
Define the estimation error of our tracking algorithm at instant time $n$ as:
\begin{equation*}
    \tilde{e}(n) := \|\hat{\w}(n)-\w(n)\|_{2}
\end{equation*}
We can decompose the expected instantaneous error into two parts by triangle inequality:
$$
\begin{aligned}
    \tilde{e}(n) & =  \| \hat{\w}(n)-\tilde{\w}(n) + \tilde{\w}(n) -\w(n) \|_{2} \\
    & \leq  \| \hat{\w}(n)-\tilde{\w}(n)\|_{2} + \|\tilde{\w}(n) -\w(n)\|_{2} \\ 
\end{aligned}
$$
where $\tilde{\w}(n)$ is one of the stationary point of the objective in Eq. (\ref{eq:ob_fc_mcp}). Hence, $\| \hat{\w}(n)-\tilde{\w}(n)\|_{2}$ is the error produced by EM step in our tracking algorithm, i.e., \textit{EM error}. $\| \tilde{\w}(n) - \w(n)\|_{2}$ is the error produced by the convex or nonconvex relaxation, i.e., \textit{relaxation error}. \\
~\\
\noindent \textbf{EM Error:} The EM update in SPALS-MCP can be considered as a composite of $\mathcal{E}: \mathbb{C}^{M} \xrightarrow{} \mathbb{C}^{M}$ and $\mathcal{T}: \mathbb{C}^{M} \xrightarrow{} \mathbb{C}^{M}$ as follows: 
\begin{equation*}
    \mathcal{E}(\w) := (\bold{I}-(\xi^2 /\sigma^2)\X^{H}\Lam\X)\w +  (\xi^2 /\sigma^2)\X^{H}\Lam\d
\end{equation*}
\begin{equation*}
    \mathcal{T}(\w) := [t_1,\cdots,t_M]^{T}
\end{equation*}
with $t_{i} := \operatorname{prox}_{\sigma^2\gamma \text{,} \rho_\alpha}(w_{i})$ defined in section \RN{2}-B. Then, $\mathcal{M}: \mathbb{C}^{M} \xrightarrow{} \mathbb{C}^{M}$ as a composite mapping of applying EM algorithm is written as:
\begin{equation*}
    \mathcal{M}(\w) = \mathcal{T}(\mathcal{E}(\w))
\end{equation*}
The mapping $\mathcal{E}$ is differentiable while $\mathcal{T}$ is non-differentiable. We can define the subdifferential (also superdifferential) for mapping $\mathcal{T}$ as: 
\begin{equation*}
    \partial\mathcal{T}(\w) := \text{diag}(s_1,\cdots,s_M),
\end{equation*}
where for firm thresholding, i.e., $\sigma^{2}\gamma < \alpha$, 
\begin{equation*}
    s_i := \left\{
    \begin{tabular} {@{\enspace}l} 
      1  \hspace{33.9mm}$|w_{i}| > \alpha$ \\ 
      $1 \leq s_i \leq \frac{\alpha} {\alpha-\sigma^2\gamma}$ \hspace{11mm} $|w_{i}| = \alpha$\\
      $\frac{\alpha}{\alpha-\sigma^2\gamma}$ \hspace{14.5mm} $\sigma^{2}\gamma<|w_i|<\alpha$ \\
      $0 \leq s_{i} \leq \frac{\alpha}{\alpha-\sigma^2\gamma}$ \hspace{11mm} $|w_{i}| = \sigma^{2}\gamma$\\
      0 \hspace{32.8mm} $|w_{i}| < \sigma^{2}\gamma$
    \end{tabular}
    \right. 
\end{equation*}

\noindent and for hard thresholding, i.e., $\sigma^{2}\gamma = \alpha$ and $\sigma^{2}\gamma > \alpha$,
\begin{equation*}
    s_i := \left\{
    \begin{tabular} {@{\enspace}l}
      1 \hspace{20mm} $|w_{i}|>\sqrt{\alpha\sigma^{2}\gamma}$ \\ 
      $0 \leq s_{i} \leq 1$ \hspace{6mm} $|w_{i}|=\sqrt{\alpha\sigma^{2}\gamma}$\\ 
      0 \hspace{20mm} $|w_{i}|<\sqrt{\alpha\sigma^{2}\gamma}$.
    \end{tabular}
    \right.
\end{equation*}
From the chain rule of subdifferential \cite{boyd2004convex}, we obtain:
\begin{equation*}
    \partial \mathcal{M}(\w) = \partial(\mathcal{T}\circ \mathcal{E}(\w)) = \partial \mathcal{T}(\mathcal{E}(\w))^{H}(\bold{I}-(\xi^2 /\sigma^2)\X^{H}\Lam\X)
\end{equation*}
~\\
\indent Denote $\rho_{\text{min}}(n)$ as the smallest eigenvalue of $\X^{H}(n)\Lam(n) \X(n)$.
If firm thresholding is used in $\mathcal{T}$, the composite mapping $\mathcal{M}$ will be Lipschitz continuous, i.e., 
\begin{align}
\label{lipschitz_ineq_1}
    \| \mathcal{M}(\hat{\w}^{(k)}(n)) - \mathcal{M}(\tilde{\w}(n)) \|_{2} \leq
    C \| \hat{\w}^{(k)}(n) - \tilde{\w}(n) \|_{2},
\end{align}
where 
\begin{equation*}
    C = \|\partial\mathcal{M}(\w)\|_{2} = \left[\alpha/(\alpha-\sigma^2\gamma)\right] \left[1-(\xi^2/\sigma^2)\rho_{\text{min}}(n) \right ].
\end{equation*}

At the beginning of stage $k+1$ of the EM algorithm, we have $\mathcal{M}(\hat{\w}^{(k)}(n)) = \hat{\w}^{(k+1)}(n)$. Since $\tilde{\w}(n)$ is a stationary point of the objective function in Eq. (\ref{eq:ob_fc_mcp}), we have $\mathcal{M}(\tilde{\w}(n)) = \tilde{\w}(n)$. Hence, the inequality in (\ref{lipschitz_ineq_1}) becomes: 
\begin{align}
\label{lipschitz_ineq_2}
    \|\hat{\w}^{(k+1)}(n) - \tilde{\w}(n)\|_{2} \leq C \| \hat{\w}^{(k)}(n) - \tilde{\w}(n) \|_{2}. 
\end{align}
After $K$ EM iterations, we have:
\begin{align}
\label{EM_update_bound}
    \|\hat{\w}^{(K)}(n) - \tilde{\w}(n)\|_{2} \leq
    C^{K} \| \hat{\w}^{(0)}(n) - \tilde{\w}(n) \|_{2}. 
\end{align}
If $C < 1$, the EM error will be sufficiently small for large $K$. The value of $C$ can be controlled by proper selection of values $\alpha$ and $\gamma$. \\

\textit{Remark 3:} If hard thresholding is used in $\mathcal{T}$, the composite mapping $\mathcal{M}$ will not be Lipschitz continuous and the above derivation does not apply. The simulation examples in Section. \RN{6} suggest that firm thresholding gives better performance than hard thresholding.

~\\
\noindent \textbf{Relaxation Error:} In the SPALS-MCP algorithm, the limit points of EM update end up with stationary points due to the continuity and non-convexity of the penalized log-likelihood function in Eq. (\ref{eq:ob_map}). We would like to obtain an upper bound of the distance between the stationary points and the global optimal, i.e., $\|\tilde{\w}(n)-\w(n)\|_{2}$. 

A class of coordinate-separable regularizers, called \textit{amenable regularizers}, was proposed by Loh and Wainright in \cite{loh2017statistical}. The MCP regularizer is a member of this class. By imposing the restricted strong convexity (RSC) condition and additional parameter constraints, a upper bound of the relaxation error $\|\tilde{\w}(n)-\w(n)\|_{2}$ can be obtained \cite{JMLR:v16:loh15a}.

The optimization problem in Eq. (\ref{eq:ob_fc_mcp}) and (\ref{eq:w_post}) is equivalent to:
\begin{equation*}
\label{eq:ob_fc_mcp_equi}
\min_{\w(n)} \frac{1}{2n}\| \Lam^{1/2}(n)\d(n)-\Lam^{1/2}(n)\X(n) \w(n)\|_{2}^{2} + \frac{\sigma^2 \gamma}{n}\rho_{\alpha}(\w(n))
\end{equation*}

\noindent Define the squared loss from above as:
\begin{equation*}
\label{eq:def_Ln}
    \mathcal{L}_{n}(\w(n)) := \frac{1}{2n}\| \Lam^{1/2}(n)\d(n)-\Lam^{1/2}(n)\X(n) \w(n)\|_{2}^{2}
\end{equation*}

\noindent $\mathcal{L}_{n}(\w(n))$ satisfies the RSC condition if 
$$
\begin{aligned}
    \langle \ \nabla \mathcal{L}_{n}(\w+\Delta) & - \nabla(\w) ,\Delta \rangle \  = \frac{1}{n} \|\Lam^{1/2}(n)\X(n) \Delta\|_{2}^{2} \\ 
    & \geq 
    \begin{cases} \alpha_{1}\|\Delta\|^{2}_{2} - \tau_{1}\frac{\log p}{n} \|\Delta\|_{1}^{2}, & \;\forall \|\Delta\|_{2} \leq 1 \\ 
    \alpha_{2}\|\Delta\|_{2} - \tau_{2}\sqrt{\frac{\log p}{n}} \|\Delta\|_{1}, & \; \forall \|\Delta\|_{2} \geq 1\end{cases}
\end{aligned}
$$
where $\alpha_{1}, \alpha_{2} > 0$ and $\tau_{1}, \tau_{2} \geq 0$ are constants. 

In our case, 
\begin{equation*}
    \nabla^{2} \mathcal{L}_{n}(\w) = \frac{1}{n}\X^{H}(n)\Lam(n)\X(n)
\end{equation*}
which is positive definite with probability 1 if each row of $\X(n)$ is i.i.d Gaussian sequence. This implies that $\mathcal{L}_{n}(n)$ is strongly convex with probability 1 and further leads to the conclusion that $\tau_{1} = \tau_{2} = 0$ and $\alpha_{1} = \alpha_{2} >0$ \cite{JMLR:v16:loh15a}.

From Theorem 1 in \cite{JMLR:v16:loh15a}, if $n$ is large enough and $\mathcal{L}_{n}$ satisfies RSC with $3\alpha < 4\alpha_{1}$ and  $\sigma^{2}\gamma \geq 4n \|\nabla \mathcal{L}_{n}(\w) \|_{\infty} = 4 \| \X^{H}(n)\Lam^{1/2}(n)\boldsymbol{\epsilon}(n) \|_{\infty}$, we then obtain the bound as follows:
$$
\begin{aligned}
    \|\tilde{\w}(n) - \w(n)\|_{2} & \leq \frac{6\sigma^2\gamma\sqrt{s}}{n(4\alpha_{1}-3\alpha)} \\
    & = \frac{6\sqrt{s}}{n(4\alpha_{1}-3\alpha)}4\|\X^{H}(n)\Lam^{1/2}(n)\boldsymbol{\epsilon}(n)\|_{\infty}
\end{aligned}
$$
with $s = \|\w(n)\|_{0}$. \\

\noindent \textbf{Bound for Estimation Error for SPALS-MCP:} Combining the results from EM and relaxation error, we can obtain the upper bound of estimation error. The following theorem states the bounding results for large enough number of measurements $n$ and appropriate choices of $\gamma$ and $\alpha$.
\begin{theorem}
Denote the estimation error of SPALS-MCP algorithm at instant time $n$ as $\tilde{e}_{\text{MCP}}(n)$, the largest and smallest eigenvalue of $\X^{H}(n)\Lam(n)\X(n)$ as $\rho_{\text{max}}(n)$ and $\rho_{\text{min}}(n)$, respectively. 

Assume the quadratic loss function $\mathcal{L}_{n}(\w(n))$ satisfies the RSC condition with parameter $\alpha_{1}$. Then at a particular time index $n_{0}$ that is large enough, with $\gamma$ and $\alpha$ being
$$
\begin{aligned}
    \frac{4\|\X^{H}(n_0)\Lam^{1/2}(n_0)\boldsymbol{\epsilon}(n_0)\|_{\infty}}{\sigma^{2}} \leq & \gamma< \frac{\alpha}{\sigma^{2}}*\frac{\rho_{\text{min}}(n_0)}{\rho_{\text{max}}(n_0)}\\
    \alpha < & \frac{4}{3}\alpha_{1},
\end{aligned}
$$
we have
\begin{equation*}
    \tilde{e}_{\text{MCP}}(n_0) \leq \frac{6\sqrt{\|\w(n_0)\|_{0}}}{ n_0(4\alpha_{1}-3\alpha)}4\|\X^{H}(n_0)\Lam^{1/2}(n_0)\boldsymbol{\epsilon}(n_0)\|_{\infty}
\end{equation*}
if we perform sufficient number of EM iterations at each instant time. 
\end{theorem}

\begin{proof}
At a particular time index $n_0$ that is large enough, the estimation error will only consist of EM and relaxation error. 

From the above analysis, we require $\sigma^2 \gamma \geq  4 \| (\X^{H}(n_0)\Lam^{1/2}(n_0)\boldsymbol{\epsilon}(n_0)) \|_{\infty}$ and $3\alpha < 4\alpha_1$ to bound the relaxation error. As for controlling the EM error along EM iterations, we also require $C < 1$ so that the upper bound of EM error can become sufficiently small if a large enough number of EM iterations can be performed at each instant time. Hence,
$$
\begin{aligned}
\left[\alpha/(\alpha-\sigma^2\gamma)\right]\left[1-(\xi^2/\sigma^2)\rho_{\text{min}}(n_0) \right ] < 1 & \\
\sigma^2 \gamma < \alpha \rho_{\text{min}}(n_0) \frac{\xi^2}{\sigma^2} &
\end{aligned}
$$
The requirement of $\xi^2 \leq \sigma^2/\rho_{\text{max}}(n_0)$ further implies
$$
\sigma^2\gamma < \alpha \frac{\rho_{\text{min}}(n_0)}{\rho_{\text{max}}(n_0)}
$$
Combining the above parameter constraints with the conclusion of Theorem 1 in \cite{JMLR:v16:loh15a} completes the proof. 
\end{proof}

Under proper choices of hyperparameters and large enough EM iterations at each time index, Theorem 4.2 shows that we can get sufficiently close to a stationary point whose distance to the true tap-weight vector can be upper bounded. Besides, the above analysis regarding estimation error only consists of EM and relaxation error without involving the variation of the tap-weight vector between two consecutive time indices. 

\section{Sparse RLS with MCP}
The previous two sections proposed and analyzed SPALS- MCP algorithm, which solves Eq. (\ref{eq:ob_fc_mcp}) at a particular time index. In real applications, we mostly deal with streaming data and time-varying tap-weight vectors. This requires us to develop online algorithms for adaptive estimation. The
following Algorithm 2, based on SPALS-MCP, is designed in a recursive manner to estimate the tap-weight vector at every time index. 

\begin{algorithm}
\caption{Sparse Recursive Least Squares with MCP (SPARLS-MCP)}
\SetAlgoLined
    
 \textbf{1. Initialization}: \\
 $\B(1) = \bold{I} - (\xi^2 /\sigma^2)\x(1)\x^{H}(1)$ \\
 $\boldsymbol{\mu}(1) = (\xi^2 /\sigma^2)\x(1)\overbar{d}(1)$ and $\hat{\w}(1) = \textbf{0}$ \\
 ~\\
 \textbf{2. Update}: \\
\textbf{For input $\x(i)$ and $d(i)$ at time index $i \; (2 \leq i\leq n)$ :} \\
 \hspace{3mm} $\B(i) = \lambda \B(i-1) - (\xi^2 /\sigma^2)\bold{x}(i)\x^{H}(i) + (1-\lambda)\bold{I}$; \\
 \hspace{3mm} $\boldsymbol{\mu}(i) = \lambda \boldsymbol{\mu}(i-1) + (\xi^2 /\sigma^2)\x(i)\overbar{d}(i)$; \\
 \hspace{3mm} $\r(i) = \B(i) \hat{\w}^{(k)}(i)+ \boldsymbol{\mu}(i)$; \\
 \hspace{3mm} $\hat{\w}^{(0)}(i) = \hat{\w}(i-1)$; \\
 ~\\ 
 \hspace{3mm} Perform the EM step in Algorithm 1 to obtain $\hat{\w}(i)$\\
 ~\\


 \textbf{3. End}: Return $\hat{\bold{w}}(n)$

\end{algorithm}

Within the loop for time index $i$ in step 2, it estimates $\w(i)$ based on $\hat{\w}(i-1)$, $\x(i)$ and $d(i)$. Notice that we recursively update the intermediate terms $\bold{B}(.)$ and $\boldsymbol{\mu}(.)$ when new signal arrives. This avoids the tedious matrix computation for $\r(.)$ in Eq. (\ref{em-e}). 

To study the convergence and steady-state error of SPARLS-MCP, we need to involve the variation of tap-weight vectors between two consecutive time indices. If the true value of tap-weight vectors do not change too rapidly in consecutive time indices, we expect that the estimate from time index $i-1$ to be a good initialization for the EM algorithm at time index $i$ for faster convergence. Numerical studies presented in Sec. \RN{6} show that our SPARLS-MCP does in fact converge relatively fast. However, mathematical proof of the convergence is beyond the scope of this paper. We anticipate that such a proof requires leveraging the proof techniques in \cite{masse2020convergence}. 


Based on the above practical consideration, all the following experiments generate tap-weight vectors under some given mechanisms, some of which follows a particular random process. These experiments are also originated from real-life scenarios.

\section{Simulation Studies}
In this section, we study simulations on three application scenarios based on which we deal with sparse system identification problem. Consider the channel $d(i) = \w^{H}(i)\x(i) + \epsilon(i) \; (1\leq i \leq n)$ with the aggregate power of each input signal $\x(i)$ being 1 and Gaussian noise variance being $\sigma^2$

Then the signal-noise-ratio (SNR) can be calculated as: $\mathbb{E}( \|\w\|_{2}^{2} )/\sigma^2$. The tracking performance is measured in terms of the normalized mean squared error (NMSE): $\mathbb{E}(\|\hat{\w}-\w\|_{2}^{2})/\mathbb{E}(\|\w\|_{2}^{2})$. We compare the performance of SPARLS-MCP with $\ell_1$-norm regularized RLS (SPARLS-$\ell_1$) and conventional RLS. 

\subsection{Jake's Model}
In Jake's model, the nonzero elements in the tap-weight vector is a sample path of a Rayleigh random process with the autocorrelation between any two time indices being characterized by Bessel function. This model is widely used in wireless communication when the transmitted signal is reflected, diffracted or attenuated by the surrounded buildings in a city center \cite{chizhik2003multiple,jakes1994microwave}. The detailed generation process of any nonzero tap-weight components is as follows: 
\begin{equation*}
    w_{i} =\sqrt{g_{c}(i)^2+g_{s}(i)^2}
\end{equation*}
 $$\text{where }g_{c}(i) = \sqrt{\frac{2}{n}} \sum_{k=1}^{n} \cos \left(A_{k}\right) \text{ and } g_{s}(i) =\sqrt{\frac{2}{n}} \sum_{i=k}^{n} \sin \left(A_{k}\right).$$ $A_{k}=2\pi f_{d} i \cos \alpha_{k}+\phi_{k}$ with i.i.d. samples $\alpha_{k}$ and $\phi_{k}$ from uniform distribution $\text{Uni}(-\pi, \pi)$. In terms of the physical meaning, $\alpha_{k}$ is the angel of the $k^{\text{th}}$ incoming wave and $\phi_{k}$ represents the random initial phase of the $k^{\text{th}}$ propagation path. The generation mechanism results in the autocorrelation function between instant time $n$ and $n+\tau$ of both $g_c(.)$ and $g_s(.)$ being $J_{0}(2\pi \tau f_d)$, where $J_0$ is the zeroth order Bessel function and $f_d$ is the Doppler frequency shift \cite{xiao2002second}.

In terms of the output signal, it is generated through the linear channel: $d(i) = \w^{H}(i)\x(i) + \epsilon(i) \; (1\leq i \leq 1000)$. And in our experiments, we set the length of tap-weight vector to be 100 with 5 components being nonzero. At instant time 501, one of the nonzero weights will be muted and one of the zero weights will be activated.

The tracking performances of different algorithms under 20dB and 30dB SNR are displayed in Fig. \ref{jakes_mse}. Both the sparse version RLS algorithms outperform the conventional RLS significantly. And the SPARLS-MCP has additional 2dB gain over the SPARLS-$\ell_1$.  

\begin{figure}[h!]
    \centering
    \includegraphics[width=8.5cm]{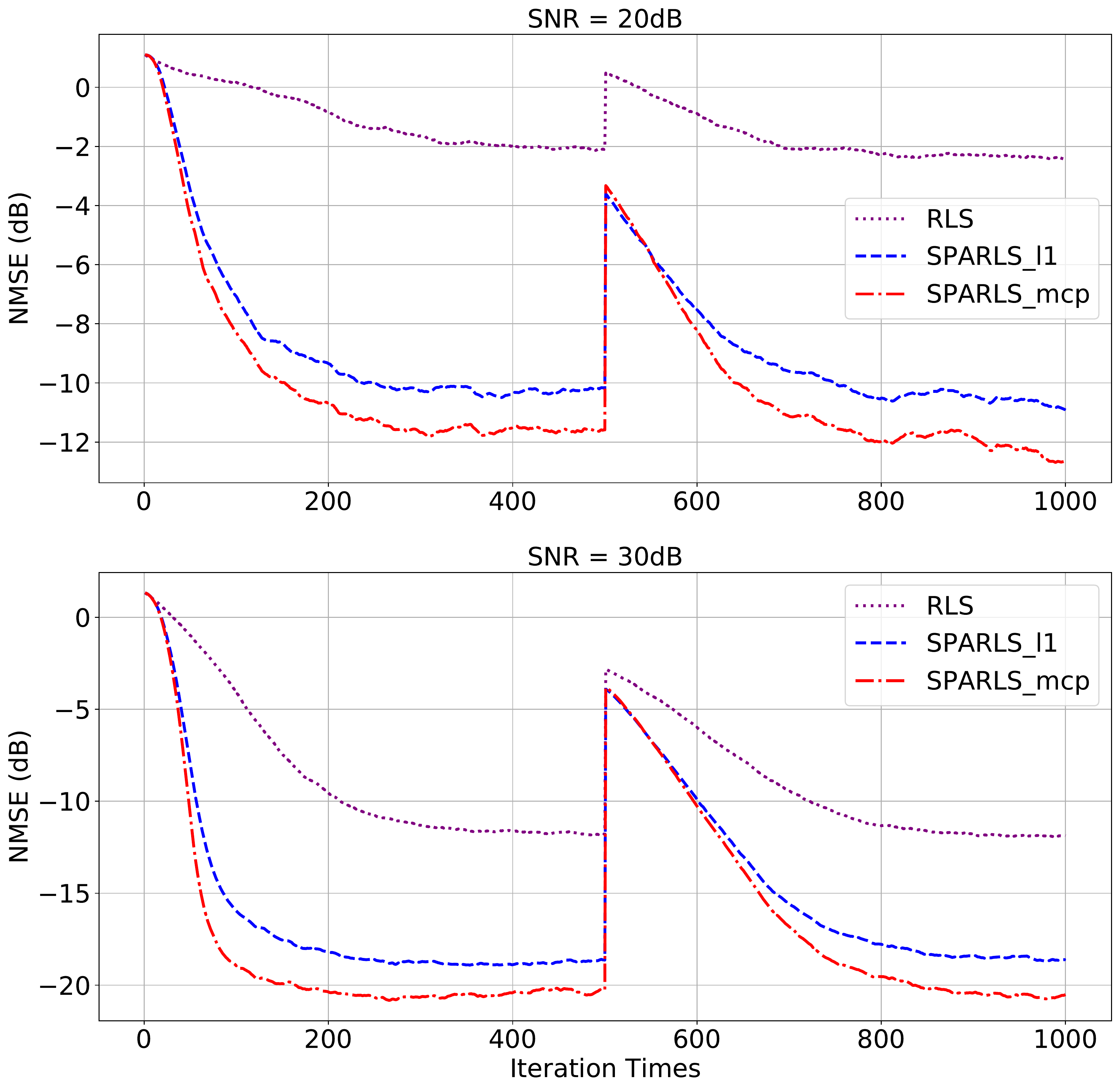}
    \caption{Comparison of tracking performance in Jake's model with $f_{d} = 0.0001$ under SNR = 20dB (upper figure) and 30dB (bottom figure)} 
    \label{jakes_mse}
\end{figure}

\subsection{Volterra System}
As a polynomial representation of a nonlinear system, Volterra series is widely used in signal processing and system identification to characterize the higher order terms in input signals. For instance, in digital satellite communication, the equipped amplifiers on both earth station and satellite repeaters are operating near saturation. Besides, only a restricted portion of bandwidth is available. These characteristics create a nonlinear channel with memory, which can be modeled by Volterra series \cite{benedetto1979modeling}. 

The discrete version of Volterra series is as follows:
\begin{equation*}
    d(i) = \sum_{p=1}^{P}\sum_{\tau_{1} = 0}^{m_p}\cdots \sum_{\tau_{p} = 0}^{m_p}h_{p}(\tau_{1},\cdots,\tau_{p})\prod_{j=1}^{p} x(i-\tau_{j}) + \epsilon(i)
\end{equation*}

\noindent where $P$ is the order of nonlinear terms, $m_p$ is the order of memory and $h_{p}(\tau_{1},\cdots,\tau_{p})$ is the Volterra kernel that needs to be estimated. In band-limited communications, only odd order products of the input signal are incorporated and the nonlinear system can be represented as \cite{kalouptsidis2011adaptive}: 
\begin{equation*}
    \begin{aligned}
    d(i)  = & \sum_{p=0}^{\lfloor \frac{(P-1)}{2} \rfloor}\sum_{\tau_{1} = 0}^{m_{2p+1}}\cdots \sum_{\tau_{2p+1} = 0}^{m_{2p+1}}h_{2p+1}(\tau_{1},\cdots,\tau_{2p+1}) \\
    & \prod_{j=1}^{p+1} x(i-\tau_{j})\prod_{j^{*} = p+2}^{2p+1}\overbar{x}(i-\tau_{j^{*}}) + \epsilon(i)
    \end{aligned}
\end{equation*}

In our experiment, we consider the following sparse third order Volterra system with memory 7 and only odd order products included: 
\begin{equation*}
    \begin{aligned}
    d(i) & = a_{3}x(i-3) + a_{5}x(i-5) + b_{1,4}x^{2}(i-1)\overbar{x}(i-4) \\
    & + b_{5,1}x^{2}(i-5)\overbar{x}(i-1) + \epsilon(i) \;\; (1 \leq i \leq 500)
    \end{aligned}
\end{equation*}

\begin{equation*}
    \begin{aligned}
    d(i) & = a^{'}_{3}x(i-3) + a^{'}_{7}x(i-7) + b^{'}_{1,4}x^{2}(i-1)\overbar{x}(i-4) \\
    & + b^{'}_{6,1}x^{2}(i-6)\overbar{x}(i-1) + \epsilon(i) \;\; (501 \leq i \leq 1000)
    \end{aligned}
\end{equation*}
where the i.i.d input signal $x(i-k) \;(0\leq k \leq 7)$ follows a complex normal distribution $\mathcal{CN}(0,1)$. The coefficients we would like to estimate, $a(.)$, $a^{'}(.)$, $b(.)$ and $b^{'}(.)$, and the i.i.d noises $\epsilon(i)$ are all generated from complex normal $\mathcal{CN}(0,1)$. At instant time 501, both the coefficients and the locations of nonzero tap-weights change. 

Since the memory of the Volterra system is 7, the first order terms incorporate $x(i-k) \;(0\leq k \leq 7)$. And we confine the third order part to be the composition of a square and first order terms, i.e., $x^{2}(i-m)\overbar{x}(i-n) \; (0\leq m \leq 7,\; 0\leq n \leq 7)$. Hence, in our sparse Volterra system, 4 out of 72 tap-weight components are active at every instant time. 

\begin{figure}[h!]
    \centering
    \includegraphics[width=8.5cm]{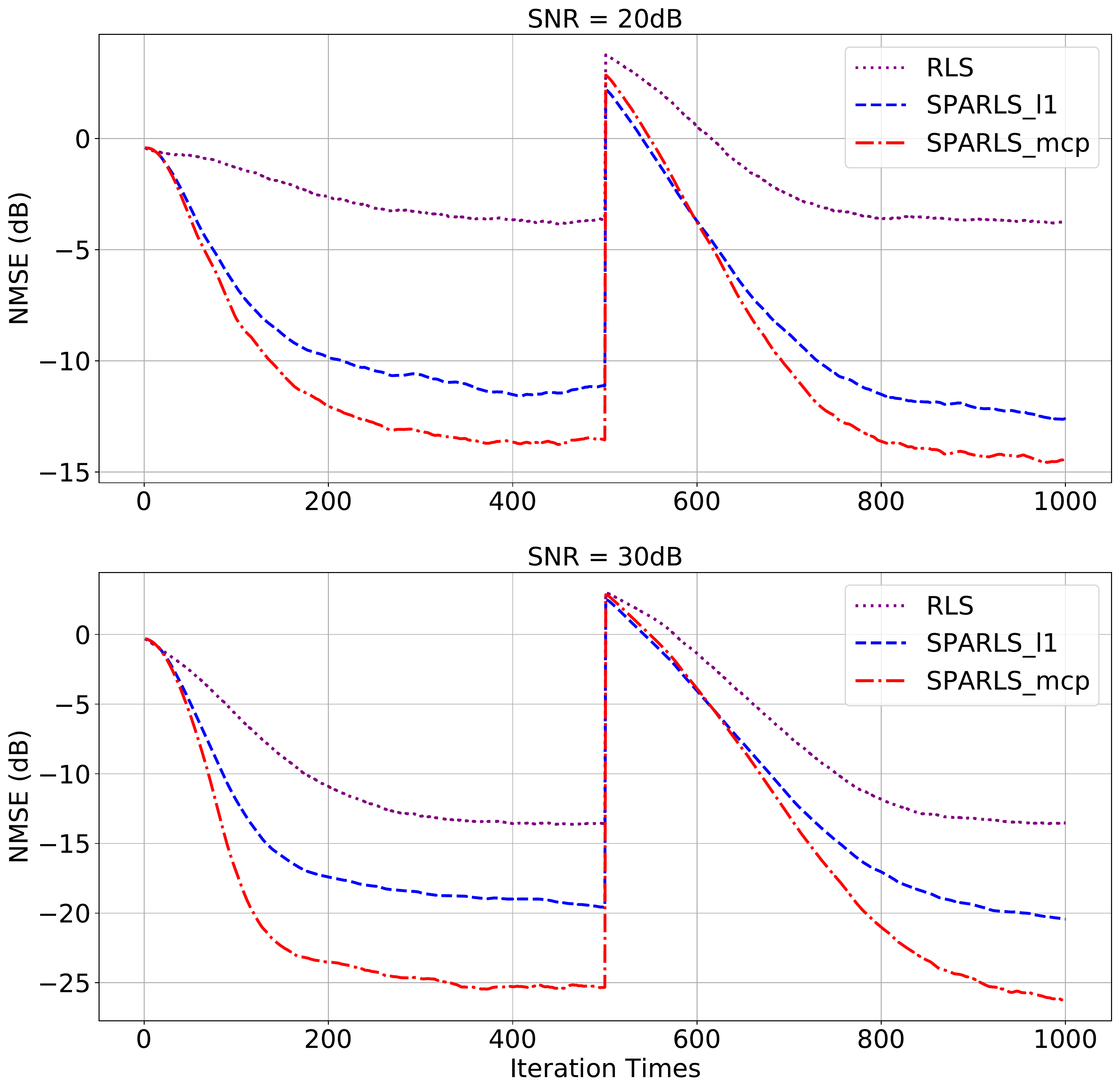}
    \caption{Comparison of tracking performance in Volterra model under SNR = 20dB (upper figure) and 30dB (bottom figure)} 
    \label{volterra_mse}
\end{figure}

The tracking performances of different algorithms under 20dB and 30dB SNR are displayed in Fig. \ref{volterra_mse}. Both the sparse version RLS algorithms outperform the conventional RLS significantly. And the SPARLS-MCP has 2dB to 5dB gain over the SPARLS-$\ell_1$. The performance gain is more obvious if the signal quality is better. 

\subsection{Multivariate Time Series}
In many real-life applications, the quantity of some index at each instant time is determined by several other factors. This is when we need multivariate time series to model the interdependency among different variables. In multivariate time series, we have more than one time-dependent variable where each variable may depend on the past values from both itself and other variables. Multivariate time series is widely used in stock market analysis, inventory study, quality control, etc \cite{heckert2002handbook}.

In our experiments, we consider the following two dimensional multivariate time series with lag 8. The goal is to forecast series $X_{2,t}$ at each instant time. 
$$
\begin{aligned}
& X_{1,t} = \epsilon_{1,t} \\
& X_{2,t} = 0.4X^{2}_{1,t-2} - 0.8X_{1,t-7} + 0.2 \epsilon_{2,t} \; (t = 1,\cdots,1000) \\
& \text{where } \epsilon_{1,t} ,\epsilon_{2,t} \text{ are i.i.d standard Gaussian}
\end{aligned}
$$

Different from the Volterra system in which we know the detailed nonlinearity information, we do not have prior knowledge about the exact forms of nonlinear terms in multivariate time series. Hence, we need to employ spline technique from statistics literature to approximate the unknown nonlinear terms. 

Assume an additive model on the series we need to forecast:
$$
\begin{aligned}
    X_{2,t} & = f(X_{1,t-1},X_{2,t-1},\cdots,X_{1,t-8},X_{2,t-8}) + \epsilon(t) \\
    & = \sum_{i=1}^{8}\left[ f_{i}^{(1)}(X_{1,t-i}) + f_{i}^{(2)}(X_{2,t-i}) \right] + \epsilon(t)
\end{aligned}
$$
with 
\begin{equation*}
    \mathbb{E}(f_{i}^{(1)}(X_{1,t-i})) = \mathbb{E}(f_{i}^{(2)}(X_{2,t-i})) = 0 \;\; (1\leq i \leq 8)
\end{equation*}
For each $f_{i}^{(1)}$ and $f_{i}^{(2)}$, we assume there are $v$ quadratic spline basis, i.e.,
$$
\begin{aligned}
    f_{i}^{(k)}(x) & = \sum_{j=1}^{v}c^{(k)}_{i,j}b^{(k)}_{i,j}(x) \\
    b_{i,j}^{(k)}(x) & = B(x|s_{1},s_{2},\cdots,s_{v-1})
\end{aligned}
$$
where $s_{1},s_{2},\cdots,s_{v-1}$ are equally spaced knots and $c_{i,j}^{(k)}$ are spline coefficients associated with basis function. 
Hence, the weighted mean squared error from instant time 1 to $n$ will be: 
\begin{equation*}
    \sum_{t=1}^{n} \lambda^{n-t}(X_{2,t}-\sum_{k=1}^{2}\sum_{j=1}^{v}\sum_{i = 1}^{8} c^{(k)}_{i,j}b^{(k)}_{i,j}(X_{k,t-i}))^{2} 
\end{equation*}
Considering the fact that $X_{2,t}$ only depends on $X_{1,t-2}$ and $X_{1,t-7}$, $f(.)$ will display sparsity with only a few $f^{(.)}_{i}(.)$ being active. Besides, the spline coefficients associated with $X_{1,t-2}$ and $X_{1,t-7}$ should be nonzero, which further implies a group sparsity pattern among all the spline coefficients. Therefore, we need to incorporate a group sparsity promoting term $\rho$ in our objective function and the problem of interest boils down to minimizing: 
\begin{equation}
\label{eq:ob_fc_mt}
  \frac{1}{2}\| \Lam^{1/2}(n)\d(n)-\Lam^{1/2}(n)\X(n) \w(n)\|_{2}^{2} + \sum_{l=1}^{L}\rho(\|\w_{l}(n)\|_{2},\gamma)
\end{equation}

\noindent where the $t^{\text{th}}$ row of $\X(n)$ is 
$$
\begin{aligned}
    \x(t) = & \: [b^{(1)}_{1,1}(X_{1,t-1}),\cdots,b^{(1)}_{1,v}(X_{1,t-1}),\\
    & \: b^{(2)}_{1,v}(X_{2,t-1}),\cdots,b^{(2)}_{1,v}(X_{2,t-1}), \\
    & \: \cdots \cdots,\\
    & \: b^{(1)}_{8,1}(X_{1,t-8}), \cdots,  b^{(1)}_{8,v}(X_{1,t-8}), \\
    & \: b^{(2)}_{8,1}(X_{2,t-8}), \cdots,  b^{(2)}_{8,v}(X_{2,t-8})],
\end{aligned}
$$
$$
\begin{aligned}
    \Lam^{1/2}(n) = & \: \text{diag}(\sqrt{\lambda^{n-1}},\sqrt{\lambda^{n-2}},\cdots,1), \\ 
    \d(n) = & \: [X_{2,1},X_{2,2},\cdots,X_{2,n}]^{T}, \\ 
    \w(n)  = & \: [ c^{(1)}_{1,1}(n),\cdots,c^{(1)}_{1,v}(n), \\
    & \: c^{(2)}_{1,1}(n),\cdots,c^{(2)}_{1,v}(n), \\
    & \: ,\cdots \cdots, \\
    & \: c^{(1)}_{8,1}(n),\cdots,c^{(1)}_{8,v}(n) \\
    & \: c^{(2)}_{8,1}(n),\cdots,c^{(2)}_{8,v}(n)]^{T}, \\ 
\end{aligned}
$$
and $L = 2\times 8 = 16$.

As for the $\w_{l}(n)$ in $\rho$, it stands for $l^{\text{th}}$ sub-vector of $\w(n)$ corresponding to the spline coefficients of $f_{\lceil l/2 \rceil}^{\tilde{l}}(.)$, where $\lceil . \rceil$ is the ceiling function and $\tilde{l} = (l \mod 2)$ and $\mod$ is the modulo operation. Specifically, 
\begin{equation*}
    \w_{l}(n) = [c_{\lceil l/2 \rceil,1}^{\tilde{l}}(n),\cdots,c_{\lceil l/2 \rceil,v}^{\tilde{l}}(n)]^{T}
\end{equation*}

We compare the performance of using group Lasso and group MCP for sparsity promoting function $\rho$ with number of quadratic spline basis $v = 10$. Fig. \ref{spline_mt} shows the estimation of spline coefficients by group lasso and group MCP, respectively. It also shows the convergence of those coefficients with the nonzero lines reaching a plateau at the end of iterations. Besides, two groups of nonzero spline coefficients are plotted in blue and orange. 

Note that the spline technique is only used for approximation purpose and there is no ground truth tap-weight vector that can be compared with. In terms of the performance comparison, we can only compare the mean prediction error, as shown in Table \ref{fitting errors}. Since the algorithm is still in the beginning stage of fitting and may not produce nearly convergent estimates during the first few hundreds iterations, only the prediction errors after instant time 400 are averaged. We can observe that both the mean and standard deviation of the prediction errors are smaller if we use group MCP regularizer. 

\begin{figure}[h!]
    \centering
    \includegraphics[width=7cm]{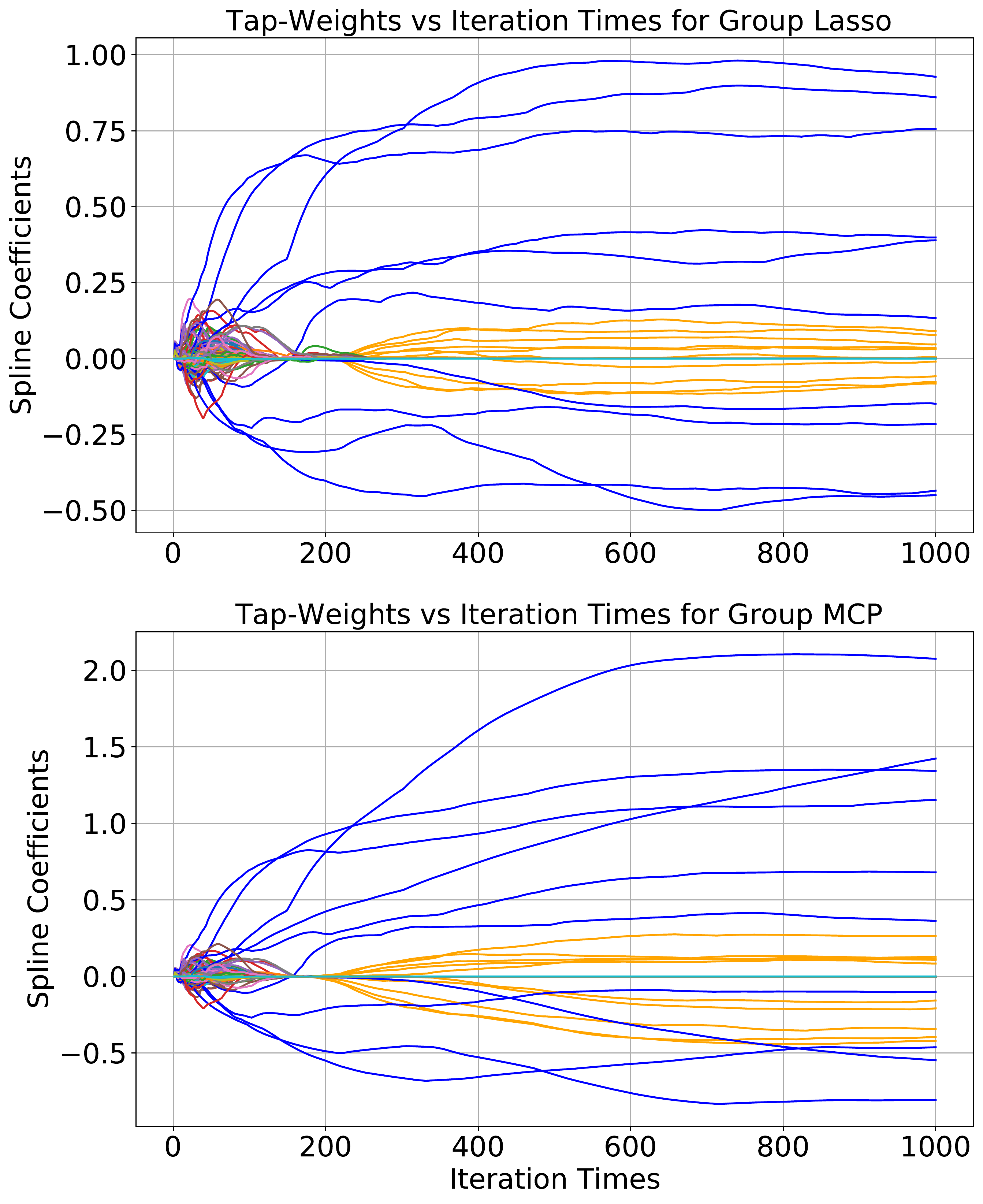}
    \caption{Spline Coefficients vs Iteration Times} 
    \label{spline_mt}
\end{figure}

\begin{table}[h!]
    \centering
    \begin{tabular}{||c c c c||} 
     \hline
      & \thead{Mean} & \thead{Std} & \thead{Quantiles (2.5$\%$, 97.5$\%$)}\\ [0.5ex] 
     \hline\hline
     Group Lasso &  \thead{0.121} & \thead{0.767} & \thead{(-0.986, 2.082)} \\ 
     Group MCP & \thead{0.040} & \thead{0.656} & \thead{(-0.693, 2.153)}\\ [1ex] 
     \hline
    \end{tabular}
    \caption{Comparison of prediction error - summary statistics (from time index 400 to 1000) }
    \label{fitting errors}
\end{table}

\subsection{Parameter Tuning}
Among all the algorithms we implement, there are three parameters that play important roles in determining our tracking performance, which are forgetting factor $\lambda$, penalization level $\gamma$ and Moreau envelope height $\alpha$.

For RLS, both theoretical analysis and simulation studies suggest that $\lambda = 1$ yields the best performance in time-invariant systems. However, in our experimental settings from Section.\RN{5}-A and \RN{5}-B, the locations of nonzero tap-weights change at some instant time and we do not need to memorize the previous nonzero locations after the change. Hence, we set $\lambda = 0.99$ to dampen the influence of signals from long time ago while preserving the best performance to the greatest extent. 

For the sake of a fair comparison, we set $\lambda = 0.99$ in both SPARLS-$\ell_1$ and SPARLS-MCP algorithm. In SPARLS-$\ell_1$, we perform a grid search on $\gamma$ under different SNR. And in SPARLS-MCP, we use the same $\gamma$ as its $\ell_1$ counterpart and perform grid search on $\alpha$. These grid searches were performed on a training dataset with known $\w(n)$ but the performance curves shown in this section have been obtained on a different dataset from the same statistical model. In spite of the fact that we are likely to end up with partially optimal choices of these parameters, there is little need for extra efforts as long as the superiority of newly proposed algorithm is demonstrated. In terms of the detailed information, the parameter choices are shown in Table \ref{paras}. 

\begin{table}[h!]
    \centering
    \begin{tabular}{|c c c|} 
    \hline
    & \thead{SNR = 20dB} & \thead{SNR = 30dB} \\ 
    \hline\hline
    Jake's Model &  $\thead{\lambda = 0.99, \gamma = 10 \\ \alpha = 0.5}$ & $\thead{\lambda = 0.99, \gamma = 30 \\ \alpha = 0.5}$  \\ 
    Volterra System & $\thead{\lambda = 0.99, \gamma = 1 \\ \alpha = 0.5}$ & $\thead{\lambda = 0.99, \gamma = 5 \\ \alpha = 0.5}$ \\ 
    \hline
    \end{tabular}
    \caption{Parameter choices of forgetting factor $\lambda$, penalization level $\gamma$ and Moreau envelope height $\alpha$ in Section.\RN{5}-A and \RN{5}-B}
    \label{paras}
\end{table}

\section{Conclusion}
We developed a SPARLS-MCP algorithm for estimation of sparse tap-weight vectors in an online adaptive system. This algorithm uses EM update to recursively estimate the tap-weight vector based on streaming data and noisy observations. We presented theoretical results regarding the convergence and estimation error of the static least squares version of our algorithm. We studied the performance of the recursive version numerically. Simulation studies under various application scenarios show better performance for SPARLS-MCP compared to SPARLS-$\ell_1$ and conventional RLS in terms of MSE. Potential directions for future work include developing recursive extensions to other MCP-regularized sparse recovery methods, such as the ADMM-MCP of \cite{wang2018admm}, and comparing the performance of such algorithms for tracking time-varying sparse models. Another possibility for future work is to develop a convergence proof for Algorithm 2, possibly using ideas from \cite{masse2020convergence}. 


\begin{appendices}

\section{Derivation of Proximal Operator of Group MCP}
\noindent Consider the minimization of:
\begin{equation*}
\mathcal{L}(\w) := \frac{1}{2\xi^2}\|\r-\w\|_{2}^{2} + \gamma\sum_{l=1}^{L}\rho_{\alpha}(\|\w_{l}\|_{2})
\end{equation*}
where 
\begin{equation*}
\rho_{\alpha}(\|\w_{l}\|_{2}) = \begin{cases}\|\w_{l}\|_{2}-\frac{1}{2 \alpha} \|\w_{l}\|_{2}^{2}, & \;\|\w_{l}\|_{2} \leq \alpha \\ 
    \frac{1}{2} \alpha, & \text { o.w. }
\end{cases} \;\;(1 \leq l \leq L)
\end{equation*}
and $\w^{T} = (\w_{1}^{T},\w_{2}^{T},\cdots,\w_{L}^{T})$ with $\w_{l} \in \mathbb{R}^{p_{l}} \;(1\leq l \leq L)$.  \\
~\\
\noindent Denote $\bold{I}_{l}$ as the $l^{\text{th}}$ column block of identity matrix $\bold{I}$ corresponding to the $l^{\text{th}}$ sub-vector in $\w$. Then, the objective can be rewritten as:  
\begin{equation*}
\mathcal{L}(\w) := \frac{1}{2\xi^2}\|\sum_{l=1}^{L}\bold{I}_{l}(\r_{l}-\w_{l})\|_{2}^{2} + \gamma\sum_{l=1}^{L}\rho_{\alpha}(\|\w_{l}\|_{2})
\end{equation*}
Similar to $\w_{l}$, $\r_{l}$ represents the $l^{\text{th}}$ sub-vector in $\r$. \\
~\\
\noindent (a) When $\|\w_{l}\|_{2} > \alpha$, setting $\partial \mathcal{L}/ \partial \w_{l} = 0$ yields: 
$$
\begin{aligned}
    \frac{1}{\xi^{2}}\bold{I}_{l}^{T}(\r - \w) & = 0 \\
    \w_{l} & = \r_{l}  
\end{aligned}
$$
Hence, $\|\r_{l}\|_{2} = \|\w_{l}\|_{2} > \alpha$. 

~\\
\noindent (b) When $0 < \|\w_{l}\|_{2} \leq \alpha$, setting $\partial \mathcal{L}/ \partial \w_{l} = 0$ yields: 
$$
\begin{aligned}
    \frac{1}{\xi^{2}}(\r_{l}-\w_{l}) - \gamma \frac{\w_{l}}{\|\w_{l}\|_{2}} + \frac{\gamma}{\alpha}\|\w_{l}\|_{2}\frac{\w_{l}}{\|\w_{l}\|_{2}} & = 0 \\
    (\frac{1}{\xi^{2}} + \frac{\gamma}{\|\w_{l}\|_{2}}-\frac{\gamma}{\alpha}) \w_{l} & = \frac{1}{\sigma^{2}}\r_{l} \;\; (*)
\end{aligned}
$$

\noindent Taking norm on both sides of Eq.$(*)$ helps us obtain $\|\w_{l}\|_{2}$:
$$
\begin{aligned}
    (\frac{1}{\xi^{2}} + \frac{\gamma}{\|\w_{l}\|_{2}}-\frac{\gamma}{\alpha})  \|\w_{l}\|_{2} & = \frac{1}{\xi^{2}}\|\r_{l}\|_{2}  \\
    \|\w_{l}\|_{2} = (\frac{1}{\xi^{2}}\|\r_{l}\|_{2}-\gamma)(\frac{1}{\xi^{2}}- \frac{\gamma}{\alpha})^{-1} & = \frac{\alpha(\|\r_{l}\|_{2}-\gamma\xi^{2})}{\alpha-\gamma\xi^2}
\end{aligned}
$$

\noindent Notice that by requiring $ 0< \|\w_{l}\|_{2} \leq \alpha$, we have: 
\begin{equation*}
    \gamma\xi^{2} < \|\r_{l}\|_{2} \leq \alpha
\end{equation*}

\noindent Finally, substituting $\|\w_{l}\|_{2}$ back into Eq.$(*)$ yields: 
\begin{equation*}
    \w_{l} = \frac{\alpha}{\alpha-\gamma\xi^2}(1-\frac{\gamma\xi^2}{\|\r_{l}\|_{2}})\r_{l}
\end{equation*}

~\\
\noindent (c) When $\w_{l} = \boldsymbol{0}$, we need to have $\boldsymbol{0} \in \partial \mathcal{L}(\w_{l})|_{\w_{l} = \boldsymbol{0}}$, where \; $\partial \mathcal{L}(\w_{l})|_{\w_{l} = \boldsymbol{0}}$ is a subgradient of $\mathcal{L}(\w)$ with respect to $\w_{l}$ evaluated at $\w_{l} = \boldsymbol{0}$. 
\begin{equation*}
    \partial \mathcal{L}(\w)|_{\w_{l} = \boldsymbol{0}} = -\frac{1}{\xi^2}\bold{I}_{l}^{T}\r + \gamma \bold{v}, \;  \text{where } \{\bold{v} \in\mathbb{R}^{p_l} \;| \; \|\bold{v}\|_{2} \leq 1\}
\end{equation*}
Hence, by letting $\boldsymbol{0} = -\frac{1}{\xi^2}\bold{I}_{l}^{T}\r + \gamma \bold{v}$, we obtain 
\begin{equation*}
    \r_{l} = \gamma\xi^{2}\bold{v}
\end{equation*}
Since $\|\bold{v}\|_{2} \leq 1$, we have
\begin{equation*}
    \|\r_{l}\|_{2} \leq \gamma\xi^{2}
\end{equation*}

~\\
Aggregating the results in (a), (b) and (c) leads to:
\begin{equation*}
\label{eq:prox_group_mcp}
    \hat{\w}_{l} = \left\{
    \begin{array}{lll}
        \boldsymbol{0}, &\|\r_{l}\|_{2} \leq \gamma\xi^2 \\ 
        \frac{\alpha}{\alpha-\gamma\xi^2}(1-\frac{\gamma\xi^2}{\|\r_{l}\|_{2}})\r_{l},   & \gamma\xi^2 < \|\r_l\|_{2} \leq \alpha \\
        \r_{l},   & \|\r_{l}\|_{2} > \alpha
    \end{array}\right. \; (1\leq l \leq L)
\end{equation*}

\end{appendices}

\bibliographystyle{IEEEtran}
\bibliography{tsp}

\begin{thebibliography}{10}
\providecommand{\url}[1]{#1}
\csname url@samestyle\endcsname
\providecommand{\newblock}{\relax}
\providecommand{\bibinfo}[2]{#2}
\providecommand{\BIBentrySTDinterwordspacing}{\spaceskip=0pt\relax}
\providecommand{\BIBentryALTinterwordstretchfactor}{4}
\providecommand{\BIBentryALTinterwordspacing}{\spaceskip=\fontdimen2\font plus
\BIBentryALTinterwordstretchfactor\fontdimen3\font minus
  \fontdimen4\font\relax}
\providecommand{\BIBforeignlanguage}[2]{{%
\expandafter\ifx\csname l@#1\endcsname\relax
\typeout{** WARNING: IEEEtran.bst: No hyphenation pattern has been}%
\typeout{** loaded for the language `#1'. Using the pattern for}%
\typeout{** the default language instead.}%
\else
\language=\csname l@#1\endcsname
\fi
#2}}
\providecommand{\BIBdecl}{\relax}
\BIBdecl

\bibitem{cui2004improved}
J.~Cui, P.~A. Naylor, and D.~T. Brown, ``An improved {IPNLMS} algorithm for
  echo cancellation in packet-switched networks,'' in \emph{2004 IEEE
  International Conference on Acoustics, Speech, and Signal Processing},
  vol.~4.\hskip 1em plus 0.5em minus 0.4em\relax IEEE, 2004, pp. iv--iv.

\bibitem{naylor2006adaptive}
P.~A. Naylor, J.~Cui, and M.~Brookes, ``Adaptive algorithms for sparse echo
  cancellation,'' \emph{Signal Processing}, vol.~86, no.~6, pp. 1182--1192,
  2006.

\bibitem{bruckstein2009sparse}
A.~M. Bruckstein, D.~L. Donoho, and M.~Elad, ``From sparse solutions of systems
  of equations to sparse modeling of signals and images,'' \emph{SIAM review},
  vol.~51, no.~1, pp. 34--81, 2009.

\bibitem{li2014improved}
Y.~Li and M.~Hamamura, ``An improved proportionate normalized least-mean-square
  algorithm for broadband multipath channel estimation,'' \emph{The Scientific
  World Journal}, vol. 2014, 2014.

\bibitem{bajwa2010compressed}
W.~U. Bajwa, J.~Haupt, A.~M. Sayeed, and R.~Nowak, ``Compressed channel
  sensing: A new approach to estimating sparse multipath channels,''
  \emph{Proceedings of the IEEE}, vol.~98, no.~6, pp. 1058--1076, 2010.

\bibitem{gu2009l_}
Y.~Gu, J.~Jin, and S.~Mei, ``$l_0$ norm constraint {LMS} algorithm for sparse
  system identification,'' \emph{IEEE Signal Processing Letters}, vol.~16,
  no.~9, pp. 774--777, 2009.

\bibitem{luo2020steady}
L.~Luo and A.~Xie, ``Steady-state mean-square deviation analysis of improved
  $l_0$-norm-constraint {LMS} algorithm for sparse system identification,''
  \emph{Signal Processing}, vol. 175, p. 107658, 2020.

\bibitem{chen2009sparse}
Y.~Chen, Y.~Gu, and A.~O. Hero, ``Sparse {LMS} for system identification,'' in
  \emph{2009 IEEE International Conference on Acoustics, Speech and Signal
  Processing}.\hskip 1em plus 0.5em minus 0.4em\relax IEEE, 2009, pp.
  3125--3128.

\bibitem{haykin2008adaptive}
S.~S. Haykin, \emph{Adaptive filter theory}.\hskip 1em plus 0.5em minus
  0.4em\relax Pearson Education India, 2008.

\bibitem{SparseRLS}
B.~Babadi, N.~Kalouptsidis, and V.~Tarokh, ``{SPARLS}: The sparse {RLS}
  algorithm,'' \emph{IEEE Transactions on Signal Processing}, vol.~58, no.~8,
  pp. 4013--4025, 2010.

\bibitem{kalouptsidis2011adaptive}
N.~Kalouptsidis, G.~Mileounis, B.~Babadi, and V.~Tarokh, ``Adaptive algorithms
  for sparse system identification,'' \emph{Signal Processing}, vol.~91, no.~8,
  pp. 1910--1919, 2011.

\bibitem{han2017slants}
Q.~Han, J.~Ding, E.~M. Airoldi, and V.~Tarokh, ``{SLANTS}: Sequential adaptive
  nonlinear modeling of time series,'' \emph{IEEE Transactions on Signal
  Processing}, vol.~65, no.~19, pp. 4994--5005, 2017.

\bibitem{zhang2010nearly}
C.-H. Zhang, ``Nearly unbiased variable selection under minimax concave
  penalty,'' \emph{The Annals of statistics}, vol.~38, no.~2, pp. 894--942,
  2010.

\bibitem{wang2018admm}
H.~Wang, Z.~Shi, C.-S. Leung, and H.~C. So, ``{ADMM-MCP} framework for sparse
  recovery with global convergence,'' \emph{arXiv preprint arXiv:1805.00681},
  2018.

\bibitem{shen2019structured}
L.~Shen, B.~W. Suter, and E.~E. Tripp, ``Structured sparsity promoting
  functions,'' \emph{Journal of Optimization Theory and Applications}, vol.
  183, no.~2, pp. 386--421, 2019.

\bibitem{figueiredo2003algorithm}
M.~A. Figueiredo and R.~D. Nowak, ``An {EM} algorithm for wavelet-based image
  restoration,'' \emph{IEEE Transactions on Image Processing}, vol.~12, no.~8,
  pp. 906--916, 2003.

\bibitem{wu1983convergence}
C.~J. Wu, ``On the convergence properties of the {EM} algorithm,'' \emph{The
  Annals of statistics}, pp. 95--103, 1983.

\bibitem{kowalski2014thresholding}
M.~Kowalski, ``Thresholding rules and iterative shrinkage/thresholding
  algorithm: A convergence study,'' in \emph{2014 IEEE International Conference
  on Image Processing (ICIP)}.\hskip 1em plus 0.5em minus 0.4em\relax IEEE,
  2014, pp. 4151--4155.

\bibitem{li2015accelerated}
H.~Li and Z.~Lin, ``Accelerated proximal gradient methods for nonconvex
  programming,'' \emph{Advances in neural information processing systems},
  vol.~28, 2015.

\bibitem{boyd2004convex}
S.~Boyd, S.~P. Boyd, and L.~Vandenberghe, \emph{Convex optimization}.\hskip 1em
  plus 0.5em minus 0.4em\relax Cambridge university press, 2004.

\bibitem{loh2017statistical}
P.-L. Loh, ``Statistical consistency and asymptotic normality for
  high-dimensional robust $ m $-estimators,'' \emph{The Annals of Statistics},
  vol.~45, no.~2, pp. 866--896, 2017.

\bibitem{JMLR:v16:loh15a}
\BIBentryALTinterwordspacing
P.-L. Loh and M.~J. Wainwright, ``Regularized m-estimators with nonconvexity:
  Statistical and algorithmic theory for local optima,'' \emph{Journal of
  Machine Learning Research}, vol.~16, no.~19, pp. 559--616, 2015. [Online].
  Available: \url{http://jmlr.org/papers/v16/loh15a.html}
\BIBentrySTDinterwordspacing

\bibitem{masse2020convergence}
P.-Y. Mass{\'e} and Y.~Ollivier, ``Convergence of online adaptive and recurrent
  optimization algorithms,'' \emph{arXiv preprint arXiv:2005.05645}, 2020.

\bibitem{chizhik2003multiple}
D.~Chizhik, J.~Ling, P.~W. Wolniansky, R.~A. Valenzuela, N.~Costa, and
  K.~Huber, ``Multiple-input-multiple-output measurements and modeling in
  {Manhattan},'' \emph{IEEE Journal on Selected Areas in Communications},
  vol.~21, no.~3, pp. 321--331, 2003.

\bibitem{jakes1994microwave}
W.~C. Jakes and D.~C. Cox, \emph{Microwave mobile communications}.\hskip 1em
  plus 0.5em minus 0.4em\relax Wiley-IEEE press, 1994.

\bibitem{xiao2002second}
C.~Xiao, Y.~R. Zheng, and N.~C. Beaulieu, ``Second-order statistical properties
  of the {WSS Jakes}' fading channel simulator,'' \emph{IEEE Transactions on
  communications}, vol.~50, no.~6, pp. 888--891, 2002.

\bibitem{benedetto1979modeling}
S.~Benedetto, E.~Biglieri, and R.~Daffara, ``Modeling and performance
  evaluation of nonlinear satellite links-a {Volterra} series approach,''
  \emph{IEEE Transactions on Aerospace and Electronic Systems}, no.~4, pp.
  494--507, 1979.

\bibitem{heckert2002handbook}
N.~A. Heckert, J.~J. Filliben, C.~M. Croarkin, B.~Hembree, W.~F. Guthrie,
  P.~Tobias, J.~Prinz \emph{et~al.}, ``Handbook 151: {NIST/SEMATECH} e-handbook
  of statistical methods,'' 2002.

\end{thebibliography}



\end{document}